\documentclass[11pt,a4paper,reqno]{article}

\usepackage{graphicx}
\usepackage{amssymb}
\usepackage{epstopdf}

\usepackage{latexsym}

\usepackage{hyperref}
\usepackage{amsmath}
\usepackage{amsthm}

\usepackage[mathscr]{eucal}
\usepackage[all]{xy}
\usepackage{mathrsfs}
\usepackage{authblk}

\usepackage{amsfonts}
\usepackage{amsmath}
\usepackage{amssymb}
\usepackage[hyperpageref]{backref}
\usepackage[dvipsnames,usenames]{color}
\usepackage{color}
\usepackage{shuffle}
\usepackage{colortbl}
\usepackage{makeidx}
\usepackage[all]{xy}

\allowdisplaybreaks
\usepackage[text={169mm,245mm},centering]{geometry}

\newtheorem{theorem}{Theorem}[section]

\newtheorem{proposition}[theorem]{Proposition}

\theoremstyle{definition}

\newtheorem{definition-lemma}[theorem]{Definition-Lemma}
\newtheorem{definition-theorem}[theorem]{Definition-Theorem}

\newtheorem*{ack}{Acknowledgements}

\newcommand{\C}{\mathbb{C}}
\newcommand{\R}{\mathbb{R}}
\newcommand{\Z}{\mathbb{Z}}

\begin{document}
\title{\textbf{Monodromy of Complexified Planar Kepler Problem}}

\author{Shanzhong Sun\thanks{Partially supported by NSFC (No. 11771303), Email: sunsz@cnu.edu.cn}}
\author{Peng You\thanks{Partially supported by NSFC (No. 11771303), Email: you-peng@163.com}}

\renewcommand\Affilfont{\small}

\affil{Department of Mathematics, Capital Normal University, Beijing 100048 P. R. China}

\date{}

\maketitle

\begin{abstract}
The planar Kepler problem is complexified and we show that this holomorphic completely integrable Hamiltonian system has nontrivial monodromy.
\end{abstract}

\begin{section}
{Introduction}
\end{section}

The Kepler problem has been studied from a variety of different perspectives and many beautiful structures continue to be unraveled even recently. However there are still mysteries to be clarified. We try to approach this subject by complexifying the problem.

If a material point moves in plane $\R^2$ with coordinates $(x,y)$, it is convenient to identify $\R^2$ with the Gauss complex plane $\C$ and write it in terms of complex variable $z:=x+iy\in\C$.
It is an old and nice idea to use complex variables in the planar $N$-body problem especially in the planar Kepler problem. In fact, Cauchy developed his theory of functions of one complex variable in order to study the Kepler equation. Here by complexified mechanical problem, we mean the coordinates of the mass point are $(x,y)\in \C^2$ instead of $\R^2$.

Recently, Behtash et al (\cite{BDSSU}) argued that in the semi-classical analysis of path integrals of supersymmetric quantum field theory and quantum mechanics, both the complexification of the action and the measure of the theory are necessary to get right nonperturbative structure of the model. It is obvious and promising from this viewpoint that we need to  complexify the Newton equations in classical mechanics to get holomorphic Newton equations in complexified configuration spaces. They claim that complexified solutions including both multi-valued and even singular ones to holomorphic Newton equations are responsible for consistency of the semi-classical theory, interesting physical properties like hidden topological angles and remarkable features to the resurgent cancellation of ambiguities to get exact solutions. There is also a huge literature on the complexification of path integration for various physical motivations (see \cite{BDSSU} and references therein).

In fact this is an idea dating back at least to Balian-Bloch (\cite{BB}) as a computational tool to get the finer structures of the spectral problems of the Schr\"{o}dinger operators in quantum mechanics, and please refer to Pham (\cite{Ph}) for the further developments and mathematical justifications.

Back to $N$-body problem in celestial mechanics, Albouy and Kaloshin (\cite{AK}) proved the generic finiteness of central configurations for planar $N$-body problem by complexifying each real coordinate and using complex algebraic geometry.

Inspired by these progresses in quantum physics and mathematics, it seems natural to develop complexified classical mechanics and to investigate the holomophic Newton equations. The relation to real Newton equations and their physical implications would be interesting problems. Other than the harmonic oscillator, the first example coming to our minds is the complexified Kepler problem and this is one of the main motivations of the current work.

Once we got this idea, we check it in the mathematical literature. Unfortunately we only find the following references: Beukers-Cushman(\cite{BeC}), Bates-Cushman(\cite{BaC}) and Bender-Hook-Kooner(\cite{BHK}) with related references therein. In the latter the authors proposed through extensive numerical studies that complex classical systems can exhibit tunneling-like behavior as in quantum mechanics.

As shown by concrete examples appearing in quantum molecular spectroscopy, monodromies of singularities in completely integrable Hamiltonian systems have quantum effects (please refer to Cushman et al(\cite{CDGHJLSZ}) as an interesting example, and see also \cite{Ch} to get a flavor of the subject).

We prove that the complexified Kepler problem is a completely integrable holomorphic Hamiltonian system and the period lattice of the regular fibers of the energy-momentum map and the monodromy group are determined which is a new phenomenon compared with the traditional real Kepler problem.

The paper is organized as follows. We recall the classical planar Kepler problem with some facts used later in \S 2. The problem is complexified in \S 3, and it becomes a holomorphic Hamiltonian system. The Hamiltonian action of $SO(2,\C)$ and the geometrical properties of the energy-momentum mapping are established in \S 3. In \S 4, we derive the period lattice on each smooth fiber of the energy-momentum mapping by rewriting the holomorphic symplectic form to get holomorphic $1$-forms $\omega_1,\,\omega_2$ and considering the generalized Abel-Jacobi map.  The main result about the monodromy of the complexified planar Kepler problem  is obtained in \S 5, we show that the problem has nontrivial monodromy which is different from the classical Kepler problem. Finally, in \S 6 some possible directions to study further are discussed.

\begin{ack}
We would like to thank Yong Li, Volodya Roubtsov and David Sauzin for their valuable suggestions.

\end{ack}

\begin{section}
{Classical planar Kepler problem}
\end{section}
We recall the classical planar Kepler Hamiltonian system first. On the phase space $T_0\R^2=(\R^2\setminus\{0\})\times\R^2$ with coordinates $(x_1,x_2,y_1,y_2)$ and symplectic form $\omega=dx_1\wedge dy_1+dx_2\wedge dy_2$, consider the Kepler Hamiltonian
\[ H:T_0\R^2\rightarrow\R; (x_1,x_2,y_1,y_2)\mapsto\frac{1}{2}(y_1^2+y_2^2)-\frac{1}{\sqrt{x_1^2+x_2^2}}.\] The integral curves of the Hamiltonian vector field $X_H$ on the $T_0\R^2$ satisfy the equations
\begin{align*}
\dot{x}&=y \\
\dot{y}&=-x\|x\|^{-3},
\end{align*}
where $x=(x_1,x_2),y=(y_1,y_2),\|x\|=\sqrt{x_1^2+x_2^2}$.

The Hamiltonian vector field $X_H$ has well known integrals: the Hamiltonian, i.e., the total energy \[H(x_1,x_2,y_1,y_2)=\frac{1}{2}(y_1^2+y_2^2)-\frac{1}{\sqrt{x_1^2+x_2^2}}\] and the angular momentum \[J(x_1,x_2,y_1,y_2)=x_1y_2-x_2y_1.\]

A Hamiltonian system $(M,\omega,f_1)$ on a smooth symplectic manifold $(M,\omega)$ of dimension $2n$ is a Liouville integrable system with $n$-degrees of freedom if there are Poisson commuting functions $(f_1,...,f_n)$, that is, $\{f_i,f_j\}=0$, whose differentials are linearly independent on an open dense subset $W$ of $M$ and whose associated Hamiltonian vector fields $X_{f_i}$ are complete. By definition, $(T_0\R^2,\omega,H,J)$ is a Liouville integrable system.

\begin{section}
{Complexified planar Kepler problem}
\end{section}

In this section, we describe the complex Hamiltonian system given by the complexified planar Kepler problem and its energy-momentum map.

Let's start from the complexified planar Kepler problem. As a first step, we are required to calculate the square roots of $x_1^2+x_2^2$ in the Kepler Hamiltonian $H$. As in complex analysis we construct the Riemann surface of $\sqrt{z}$ as $\sqrt{z}$'s maximal domain to remove the multivaluedness, we "double the $(x_1,x_2)$ configuration space" to make the function $\sqrt{x_1^2+x_2^2}$ single valued and holomorphic.

We take two copies of $\C^2\setminus\{(x_1,x_2)\in\C^2:x_1^2+x_2^2=0\}$, one is denoted by $Q_{I}$ and the other by $Q_{II}$. For convenience purpose we introduce the new coordinates $\xi_1=x_1+ix_2$, $\xi_2=x_1-ix_2\in\C$. Cut $Q_I$ and $Q_{II}$ along $x_1^2+x_2^2=\xi_1\xi_2=t\in\R_{<0}$ as follows: for each fixed $\xi_2\neq 0$, we get two copies of $\C\setminus\{0\}$ in $Q_{I}$ and $Q_{II}$ respectively denoted simply by $I$ and $II$, cut both $I$ and $II$ along the half line $\xi_1=t\xi_2^{-1}$, $t<0$ to get $I_\pm$ and $II_\pm$. Glue $I_+$ with $II_-$ and $I_-$ with $II_+$ respectively as shown in Figure \ref{fig_a}.
\begin{figure}
\centering
\includegraphics[scale=0.6]{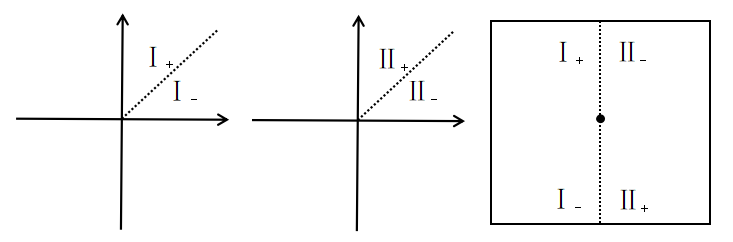}
\caption{Glue $I$ and $II$ together for each fixed $\xi_2\neq 0$.}\label{fig_a}
\end{figure}
Then let $\xi_2$ vary, we obtain the glued space $Q=Q_I\coprod Q_{II}$. In the following we will use subscript $k$ in $(x_1,x_2)_k$ to indicate which $Q_k$, $k=I,II$ the point $(x_1,x_2)$ is in.

The configuration space $Q$ is a complex manifold of dimension $2$. The phase space $M=T^*Q$ is a trivial bundle. Furthermore $M$ is a holomorphic symplectic manifold with the canonical symplectic form $\omega$ which in natural coordinate $(x_1,x_2,y_1,y_2)$ can be written as \[\omega =dx_1\wedge dy_1+dx_2\wedge dy_2.\] Now the Hamiltonian $H=\frac{1}{2}(y_1^2+y_2^2)-\frac{1}{\sqrt{x_1^2+x_2^2}}$ makes sense on $M$ and is holomorphic. The complexified planar Kepler problem becomes a holomorphic Hamiltonian system $(M,\omega,H)$.

The group \[SO(2,\C)=\{A=
\begin{pmatrix}
a & -b \\ b & a
\end{pmatrix}
\in GL(2,\C):a^2+b^2=1\}\] which preserves $H$, acts on $M$ by \[\tilde{\Phi}:SO(2,\C)\times M\rightarrow M; (A,((x_1,x_2)_k,(y_1,y_2)))\mapsto ((A
\begin{pmatrix}
x_1 \\ x_2
\end{pmatrix}
)_k, A
\begin{pmatrix}
y_1 \\ y_2
\end{pmatrix}
), k=I,II.\]
As in the usual planar Kepler problem, the action $\tilde\Phi$ is a Hamiltonian action of $SO(2,\C)$ on the holomorphic symplectic manifold $(M,\omega)$ with moment map the (complexified) angular momentum \[J(x_1,x_2,y_1,y_2)=x_1y_2-x_2y_1.\] Similar to the real case, we can prove that $(M,\omega, H,J)$ is a Liouville completely integrable system in the holomorphic sense.

It is convenient to use the coordinates $(\xi_1,\xi_2,\eta_1,\eta_2)$ on $M$ where \[\xi_1=x_1+ix_2,\ \xi_2=x_1-ix_2,\ \eta_1=y_1+iy_2,\ \eta_2=y_1-iy_2.\] Note that $\xi_1\neq 0$ and $\xi_2\neq 0$ because $\xi_1\xi_2=x_1^2+x_2^2\neq 0$. In terms of these new coordinates the $SO(2,\C)$-action $\tilde{\Phi}$ becomes the $\C^*$-action \[\Phi:\C^*\times M\rightarrow M:(\lambda,((\xi_1,\xi_2)_k,(\eta_1,\eta_2))\mapsto ((\lambda\xi_1,\lambda^{-1}\xi_2)_k,(\lambda\eta_1,\lambda^{-1}\eta_2)),\ k=I,II,\] Where $\lambda=a+ib$, $A=
\begin{pmatrix}
a & -b \\ b & a
\end{pmatrix}
\in SO(2,\C)$. The action of a group element $\lambda$ on a point $P$ is usually written as $\Phi_{\lambda}(P)$.

We have
\begin{proposition}\label{groupaction}
1. The $\C^*$-action $\Phi$ on $M$ is free and proper.

2. The map \[\pi:M\rightarrow \C^4; ((\xi_1,\xi_2)_k,(\eta_1,\eta_2)) \mapsto (c_1,c_2,c_3,c_4),k=I,II.\] is the quotient map $M\rightarrow M/\C^*$ of the $\C^*$-action $\Phi$, where
\[c_1=
\begin{cases}
\sqrt{\xi_1\xi_2}&(\xi_1,\xi_2)_I\in Q_I\\
-\sqrt{\xi_1\xi_2}&(\xi_1,\xi_2)_{II}\in Q_{II}
\end{cases}
,\ c_2=\xi_1\eta_2,\ c_3=\frac{i}{2}(\xi_1\eta_2-\xi_2\eta_1),\ c_4=\frac{1}{2}\eta_1\eta_2-\frac{1}{c_1}.\]

3. The image of the map $\pi$ is the subset $R$ of $\C^4$ defined by \[\{(c_1,c_2,c_3,c_4)\in\C^4:2c_4c_1^2+2c_1=c_2^2+2ic_3c_2, c_1\neq0\}.\]
Note that $c_4$ and $c_3$ are just the Hamiltonian and the angular momentum in terms of these new coordinates.

4. The map $\pi:M\rightarrow R$ is the projection map of a $\C^*$-bundle over $R$.
\end{proposition}
\begin{proof}
1. If $((\lambda\xi_1,\lambda^{-1}\xi_2)_k,(\lambda\eta_1,\lambda^{-1}\eta_2))=((\xi_1,\xi_2)_k,(\eta_1,\eta_2))$, then $\lambda =1$ for $\xi_1\neq 0$. Hence $\C^*$ acts freely on $M$. The action $\Phi$ is proper if and only if $L=\{\lambda\in\C^*:K\cap\Phi_\lambda (K)\neq\emptyset\}$ is compact for every compact set $K$ in $M$. Since $K$ is compact, there exist two positive numbers $\rho_1<\rho_2$ such that $\rho_1 <|\xi_1|,\ |\xi_2|<\rho_2$ on $K$. For every $\lambda\in L$, there is a point $((\xi_1,\xi_2)_k,(\eta_1,\eta_2))$ such that $((\lambda\xi_1,\lambda^{-1}\xi_2)_k,(\lambda\eta_1,\lambda^{-1}\eta_2))$ is also in $K$. Because $|\xi_1|,\ |\xi_2|<\rho_2$ and $\rho_1 <|\lambda\xi_1 |,\ |\lambda^{-1}\xi_2 |$,  we get $\frac{\rho_1}{\rho_2}<|\lambda|<\frac{\rho_2}{\rho_1}$. If $\lambda_n(\in L)\rightarrow\lambda$, then there exist $P_n\in K$ such that $\Phi_{\lambda_n}(P_n)\in K$. Since $K$ is compact, we have a subsequence of $\{P_n\}_{n\geq 1}$ converging to a point $P\in K$, we still use $\{P_n\}_{n\geq 1}$ to denote this subsequence. Because of the continuity of $\Phi$ and compactness of $K$, $\Phi_{\lambda}(P)\in K$. Then $L$ is compact.

2. From the definition of the map $\pi$, the following equalities hold for points of $M$
\begin{align*}
\xi_1\xi_2&=c_1^2;\\
\xi_1\eta_2&=c_2;\\
\xi_2\eta_1&=c_2+2ic_3;\\
\eta_1\eta_2&=2c_4+\frac{2}{c_1}.
\end{align*}
Let $(c_1,c_2,c_3,c_4)$ be in the image of the map $\pi$. If two points $((\xi_1,\xi_2)_k,(\eta_1,\eta_2))$ and $((\tilde{\xi}_1,\tilde{\xi}_2)_{\tilde{k}},(\tilde{\eta}_1,\tilde{\eta}_2))$ are mapped to the same $(c_1,c_2,c_3,c_4)$, we have $k=\tilde{k}$, $\xi_1\xi_2=\tilde{\xi}_1\tilde{\xi}_2$, $\xi_1\eta_2=\tilde{\xi}_1\tilde{\eta}_2$, $\xi_2\eta_1=\tilde{\xi}_2\tilde{\eta}_1$. Then there exists a unique $\lambda\in\C^*$ such that $\tilde{\xi}_1=\lambda\xi_1$, $\tilde{\xi}_2=\lambda^{-1}\xi_2$, $\tilde{\eta}_1=\lambda\eta_1$, $\tilde{\eta}_2=\lambda^{-1}\eta_2$. Noting that $c_i (i=1,2,3,4)$ is invariant under the $\C^*$-action $\Phi$, we see that the inverse image of $(c_1,c_2,c_3,c_4)$ under the map $\pi$ is a single $\C^*$-orbit.

3. Observe that the identity $(\xi_1\xi_2)(\eta_1\eta_2)=(\xi_1\eta_2)(\xi_2\eta_1)$ gives rise to
\begin{equation}
2c_4c_1^2+2c_1=c_2^2+2ic_3c_2,\ c_1\neq 0.\label{image}
\end{equation}
For $(c_1,c_2,c_3,c_4)$ satisfying the condition (\ref{image}), the point
\begin{equation*}
((1,c_1^2)_k,(\frac{c_2+2ic_3}{c_1^2},c_2)),\ k=
\begin{cases}
I&-\frac{\pi}{2}<\arg c_1\leq\frac{\pi}{2}\\
II&\frac{\pi}{2}<\arg c_1\leq\frac{3\pi}{2}
\end{cases}
\end{equation*}
is a preimage of $(c_1,c_2,c_3,c_4)$ under the quotient map $\pi$.

4. It is a general fact that if $\C^*$ acts freely and properly on a complex manifold $N$, then the quotient map $N\rightarrow N/\C^*$ is the projection map of a $\C^*$-bundle over $N/\C^*$. In particular, this holds when $N=M$.
\end{proof}

For any fixed value $c=(c_3,c_4)$, we consider the algebraic curve \[2c_4c_1^2+2c_1=c_2^2+2ic_3c_2.\] It is easy to check that the curve is smooth if and only if $c$ does not lie in the discriminant $\Delta$: \[\Delta:=\{(c_3,c_4)\in \C^2:1+2c_4c_3^2=0\}.\]

We slightly change the form of the algebraic curve as
\begin{equation}
2c_4c_1^2+2c_1-c_3^2=(c_2+ic_3)^2,\label{curve}
\end{equation}
which falls into one of the three cases:

1. If $c_4=0$, the curve is in the form \[2c_1-c_3^2=(c_2+ic_3)^2,\] which is homeomorphic to $\C$.

2. If $c_4\neq0$ and the discriminant of the left hand side of (\ref{curve}) $4+8c_4c_3^2\neq0$, the curve is in the form \[2c_4(c_1-r)(c_1-s)=(c_2+ic_3)^2,r\neq s,\] which is homeomorphic to $\C^*$.

3. If $c_4\neq0$ and the discriminant of the left hand side of (\ref{curve}) $4+8c_4c_3^2=0$, the curve is in the form \[2c_4(c_1+\frac{1}{2c_4})^2=(c_2+ic_3)^2,\] which is homeomorphic to a conical surface.

In the first two cases, the curve is smooth. In the third case, the curve is singular.

For any fixed value $c=(c_3,c_4)$, we define $R_c$ as the algebraic curve with points such that $c_1=0$ removed \[\{(c_1,c_2)\in\C^2:2c_4c_1^2+2c_1=c_2^2+2ic_3c_2\}\setminus\{(0,0),(0,-2ic_3)\}\] and $M_c$ as the $c$-level set of the energy-momentum map \[EM:M\rightarrow\C^2;((\xi_1,\xi_2)_k,(\eta_1,\eta_2))\mapsto (c_3,c_4),\ k=I,II,\] where $c_3$, $c_4$ are defined in Proposition \ref{groupaction}.

\begin{section}
{Period lattice}
\end{section}
In this section we obtain a period lattice for the smooth fibers of the energy-momentum map of the complexified Kepler problem.

First, we choose good coordinates on $M$ and express the symplectic form $\omega$ in these coordinates. Let $V$ be the open subset of $M$ defined by \[\{((\xi_1,\xi_2)_k,(\eta_1,\eta_2))\in M,k=I,II:c_2+ic_3\neq 0\}.\] The variables $(\xi_1,c_1,c_3,c_4)$ can be used as local coordinates at every point $p$ of $V$, because the condition $c_2+ic_3\neq 0$ implies the Jacobian of the energy-momentum map at $p$ is nonsingular. We now write the symplectic form $\omega$ in these coordinates.
\begin{proposition}
The symplectic form $\omega$ on $M$ restricted to $V$ is \[\left((\frac{i}{c_1}+\frac{-c_3}{c_1(c_2+ic_3)})dc_1+\frac{-i}{\xi_1}d\xi_1\right)\wedge dc_3+\frac{c_1}{c_2+ic_3}dc_1\wedge dc_4.\]
\end{proposition}
\begin{proof}
The canonical symplectic form $\omega$ on $M$ can be written as \[\frac{1}{2}(d\xi_2\wedge d\eta_1+d\xi_1\wedge d\eta_2).\] Since \[\xi_2=\frac{c_1^2}{\xi_1},\ \eta_1=\frac{c_2+2ic_3}{\xi_2}=\frac{(c_2+2ic_3)\xi_1}{c_1^2},\ \eta_2=\frac{c_2}{\xi_1}\] on $V$, we obtain \[d\xi_2=\frac{2c_1}{\xi_1}dc_1+\frac{-c_1^2}{\xi_1^2}d\xi_1,\] \[d\eta_1=\frac{\xi_1}{c_1^2}dc_2+\frac{2i\xi_1}{c_1^2}dc_3+\frac{c_2+2ic_3}{c_1^2}d\xi_1+\frac{-2(c_2+2ic_3)\xi_1}{c_1^3}dc_1\] and \[d\eta_2=\frac{1}{\xi_1}dc_2+\frac{-c_2}{\xi_1^2}d\xi_1.\] Hence \[\omega|_V=\frac{1}{c_1}dc_1\wedge dc_2+\frac{2i}{c_1}dc_1\wedge dc_3+\frac{-i}{\xi_1}d\xi_1\wedge dc_3.\] Differentiating $c_2^2+2ic_3c_2=2c_4c_1^2+2c_1$ gives \[dc_2=\frac{2c_4c_1+1}{c_2+ic_3}dc_1+\frac{-ic_2}{c_2+ic_3}dc_3+\frac{c_1^2}{c_2+ic_3}dc_4,\] which substituted into $\omega|_V$ yields \[\omega|_V=\left((\frac{i}{c_1}+\frac{-c_3}{c_1(c_2+ic_3)})dc_1+\frac{-i}{\xi_1}d\xi_1\right)\wedge dc_3+\frac{c_1}{c_2+ic_3}dc_1\wedge dc_4.\]
\end{proof}

Let \[\omega_1=(\frac{i}{c_1}+\frac{-c_3}{c_1(c_2+ic_3)})dc_1+\frac{-i}{\xi_1}d\xi_1,\ \omega_2=\frac{c_1}{c_2+ic_3}dc_1.\] The differential forms $\omega_1$, $\omega_2$ restricted to the subset of $M_c$ where $c_2+ic_3\neq 0$ can be extended holomorphically to $M_c$. In fact, $3$-forms $\omega_1\wedge dc_3\wedge dc_4=\omega\wedge dc_4$ and $\omega_2\wedge dc_4\wedge dc_3=\omega\wedge dc_3$ are holomorphic on $M$. We also use $\omega_1$ and $\omega_2$ to denote the extended $1$-forms.

We construct the period lattice on $M_c$ as follows. Assuming $c\notin\Delta$. Fix a point $p_0\in M_c$ and consider the $Abel$-$Jacobi$ map \[AJ:M_c\rightarrow \C^2:p\mapsto (\int_{\gamma}\omega_1,\int_{\gamma}\omega_2),\] where $\gamma$ is a path from $p_0$ to $p$ on $M_c$. Since $\omega_1$ and $\omega_2$ are closed forms, the map $AJ$ depends only on the homology class of the path on $M_c$ from $p_0$ to $p$. The map $AJ$ is multivalued  since $p_0$ and $p$ can be joined by nonhomologous paths, say $\gamma$ and $\tilde{\gamma}$. The difference in the corresponding values of $AJ$ is given by a vector $(\int_{\Gamma}\omega_1,\int_{\Gamma}\omega_2)$, where $\Gamma$ is the closed path tracing out $\gamma$ first and then tracing out $\tilde{\gamma}$ backwards. Thus the multivaluedness of $AJ$ is given by vectors $(\int_{\Gamma}\omega_1,\int_{\Gamma}\omega_2)$ where $\Gamma$ ranges over $H_1(M_c,\Z)$. Let \[L_c=\{(\int_{\Gamma}\omega_1,\int_{\Gamma}\omega_2)\in \C^2:\Gamma\in H_1(M_c,\Z)\}.\] When $c_3\neq 0$, the first homology groups of $M_c$ and $R_c$ are $H_1(M_c,\Z)=\Z^4$ and $H_1(R_c,\Z)=\Z^3$ respectively.

\begin{theorem}
Let $\gamma_1$, $\gamma_2$, $\gamma_3$ be a $\Z$-basis of $H_1(R_c,\Z)$ where $\gamma_1$ is the generator of the first homology group of the smooth algebraic curve $2c_4c_1^2+2c_1=c_2^2+2ic_3c_2$. Define \[\mu=\int_{\gamma_1}\left(\frac{i}{c_1}+\frac{-c_3}{c_1(c_2+ic_3)}\right)dc_1,\ \nu=\int_{\gamma_1}\left(\frac{c_1}{c_2+ic_3}\right)dc_1.\] Then $L_c$ is a lattice in $\C^2$ called the period lattice. $L_c$ has $\Z$-rank two and is generated by the vectors \[(\mu,\nu),\ (2\pi,0).\]
\end{theorem}
\begin{proof}
The point $(c_1,c_2)=(0,-2ic_3)$ is the removable singularity of the differential $\left(\frac{i}{c_1}+\frac{-c_3}{c_1(c_2+ic_3)}\right)dc_1$ on $R_c$, then \[\int_{\gamma_3}\left(\frac{i}{c_1}+\frac{-c_3}{c_1(c_2+ic_3)}\right)dc_1 =0.\] The point $(c_1,c_2)=(0,0)$ is the pole of the differential $\left(\frac{i}{c_1}+\frac{-c_3}{c_1(c_2+ic_3)}\right)dc_1$ on $R_c$ which is first order. If $c_3\neq 0$, the residue at the pole $(0,0)$ is $2i$, then \[\int_{\gamma_2}\left(\frac{i}{c_1}+\frac{-c_3}{c_1(c_2+ic_3)}\right)dc_1 =-4\pi.\] If $c_3=0$, the residue at the pole $(0,0)$ is $i$, then \[\int_{\gamma_2}\left(\frac{i}{c_1}+\frac{-c_3}{c_1(c_2+ic_3)}\right)dc_1 =-2\pi.\]

The map \[M_c\rightarrow R_c\times\C^*:((\xi_1,\xi_2)_k,(\eta_1,\eta_2))\mapsto (c_1,c_2,\xi_1)\] is a diffeomorphism. Let $\Gamma$ be any closed path in $M_c$. Without loss of generality we can assume that $\Gamma$ lies in the $M_c\cap V$. The image in $R_c$ of the curve $\Gamma$ under the projection $\pi$ is homologous to $n_1\gamma_1+n_2\gamma_2+n_3\gamma_3$. Hence \[\int_{\Gamma}\omega_2=\int_{\pi(\Gamma)}\frac{c_1}{c_2+ic_3}dc_1=n_1\int_{\gamma_1}\frac{c_1}{c_2+ic_3}dc_1+n_2\int_{\gamma_2}\frac{c_1}{c_2+ic_3}dc_1+n_3\int_{\gamma_3}\frac{c_1}{c_2+ic_3}dc_1=n_1\nu,\] and
\begin{align*}
\int_{\Gamma}\omega_1=&\int_{\Gamma}\frac{-i}{\xi_1}d\xi_1+\int_{\pi(\Gamma)}\left(\frac{i}{c_1}+\frac{-c_3}{c_1(c_2+ic_3)}\right)dc_1\\
=&\int_{\Gamma}\frac{-i}{\xi_1}d\xi_1+n_1\int_{\gamma_1}\left(\frac{i}{c_1}+\frac{-c_3}{c_1(c_2+ic_3)}\right)dc_1+n_2\int_{\gamma_2}\left(\frac{i}{c_1}+\frac{-c_3}{c_1(c_2+ic_3)}\right)dc_1\\
&+n_3\int_{\gamma_3}\left(\frac{i}{c_1}+\frac{-c_3}{c_1(c_2+ic_3)}\right)dc_1\\
=&n_1\mu+2m\pi,
\end{align*}
where $m$ is some integer. The last equality holds because we can let $\Gamma$ encircle $\xi_1=0$ as many times as we want. This ends the proof.
\end{proof}
%

\begin{section}
  {Monodromy}
\end{section}
In this section we show that the complexified planar Kepler problem has monodromy.  We will prove that the monodromies are due to chasing the nontrivial loops around the discriminant on the one hand and the $c_3$ axis on the other where the fiber $M_c$ changes its topology.

\subsection{Monodromy around the discriminant \texorpdfstring{$\Delta$}{Delta}}

More precisely, we will show that there is a noncontractible loop $\Gamma$ in $\C^2\setminus\Delta$ such that the bundle of period lattices \[\phi:\coprod_{c\in\Gamma}L_c\rightarrow\Gamma\] over $\Gamma$ has classifying map given by
$\begin{pmatrix}
-1 & 0 \\ -2 & 1
\end{pmatrix}$.
\begin{theorem}\label{monodromy}
Consider the periods of the differential forms $\omega_1$ and $\omega_2$ on $R_c$ as analytic functions of $c=(c_3,c_4)$. Then there is a closed path $\Gamma$ in the $c$-space $\C^2\setminus\Delta$ such that after analytic continuation along one circuit of $\Gamma$, the period of $\alpha=\left(\frac{i}{c_1}+\frac{-c_3}{c_1(c_2+ic_3)}\right)dc_1$ turns into $-4\pi-\int_\gamma\alpha$ and the period of $\omega_2$ changes the sign.
\end{theorem}
\begin{proof}
The idea of the proof is that we normalize the curve $R_c$ into $2u(u-1)=v^2$, then we find the change of the periods of $\alpha$ and $\omega_2$ along a well chosen closed path $\Gamma$ in $\C^2\setminus\Delta$.

If $(c_3,c_4)\notin\C^2\setminus\Delta$ and $c_4\neq 0$, the curve has the form \[2c_4(c_1-r)(c_1-s)=(c_2+ic_3)^2,\] where $r=\frac{-1-\sqrt{1+2c_3^2c_4}}{2c_4}$, $s=\frac{-1+\sqrt{1+2c_3^2c_4}}{2c_4}$. Taking the substitution \[c_1=(s-r)u+r,\ c_2=(s-r)\sqrt{c_4} v-ic_3\] transforms the equation $2c_4(c_1-r)(c_1-s)=(c_2+ic_3)^2$ into the equation $2u(u-1)=v^2$. Under the same transformation, the $1$-forms $\alpha$ and $\omega_2$ are changed into \[\alpha=\left(\frac{i}{u-\frac{-r}{s-r}}+\frac{-c_3}{\sqrt{c_4}(s-r)}\frac{1}{v(u-\frac{-r}{s-r})}\right)du,\ \omega_2=\left(\frac{s-r}{\sqrt{c_4}}\frac{u-\frac{-r}{s-r}}{v}\right)du\] and the pole $(0,0)$ of $\alpha$ will be transformed to the point $(\frac{-r}{s-r},\frac{ic_3}{(s-r)\sqrt{c_4}})$.

Take a point $(1,\frac{-1}{2})\in\Delta$ and a circle $\Gamma$ with the center $(1,\frac{-1}{2})$ and the radius $r_1<\frac{1}{2}$ in the complex plane $\{(c_3,c_4)\in\C^2\,|\, c_3=1\}$. Then on closed path $\Gamma$ the period of $\alpha$ and $\omega_2$ will be \[\int_\gamma\alpha=\int_\gamma\left(\frac{i}{u-(\frac{1}{2}+\frac{1}{2\sqrt{1+2c_4}})}+\frac{-\sqrt{c_4}}{\sqrt{1+2c_4}}\frac{1}{\sqrt{2u(u-1)}}\frac{1}{u-(\frac{1}{2}+\frac{1}{2\sqrt{1+2c_4}})}\right)du\] and \[\int_\gamma\omega_2=\int_\gamma\left(\frac{\sqrt{1+2c_4}}{c_4\sqrt{c_4}}\frac{u-(\frac{1}{2}+\frac{1}{2\sqrt{1+2c_4}})}{\sqrt{2u(u-1)}}\right)du.\] As shown in Figure \ref{fig_b}, the pole $\frac{1}{2}+\frac{1}{2\sqrt{1+2c_4}}$ of $\alpha$ will go along a semicircle on the circle with the radius $\frac{1}{\sqrt{r_1}}$ and the center $\frac{1}{2}$ in the $u$-plane, when we let $c_4$ go around $\Gamma$ once. By choosing $\gamma$ as a circle with the center $\frac{1}{2}$ and the radius larger than $\frac{1}{\sqrt{r_1}}$, we get \[\int_\gamma\alpha=-2\pi+\frac{-\sqrt{c_4}}{\sqrt{1+2c_4}}\int_\gamma\left(\frac{1}{\sqrt{2u(u-1)}}\frac{1}{u-(\frac{1}{2}+\frac{1}{2\sqrt{1+2c_4}})}\right)du.\] When $c_4$ goes through $\Gamma$ once, the integral $\int_\gamma\left(\frac{1}{\sqrt{2u(u-1)}}\frac{1}{u-(\frac{1}{2}+\frac{1}{2\sqrt{1+2c_4}})}\right)du$ does not change but $\frac{-\sqrt{c_4}}{\sqrt{1+2c_4}}$ changes sign, then $\int_\gamma\alpha$ turns into $-4\pi-\int_\gamma\alpha$. For the same reason, $\int_\gamma\omega_2$ changes the sign after analytic continuation.
\begin{figure}
\centering
\includegraphics[scale=0.6]{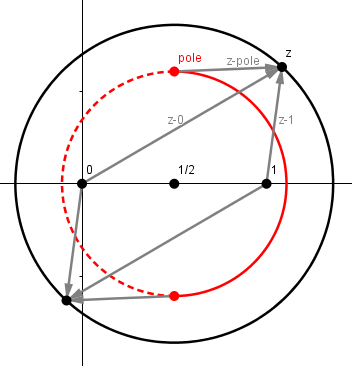}
\caption{Variation of periods in $u$-plane: encompassing the discriminant. The red curve is the loci of the poles when traveling along $\Gamma$, and the outer circle is the common cycle $\gamma$ to calculate the family of periods.}\label{fig_b}
\end{figure}
\end{proof}

\subsection{Monodromy around \texorpdfstring{$c_3$}{c3}}

In this subsection we will show that there is a noncontractible loop $\Gamma$ in $\C^2\setminus\{(c_3,c_4=0):c_3\in\C\}$ such that the bundle of period lattices \[\phi:\coprod_{c\in\Gamma}L_c\rightarrow\Gamma\] over $\Gamma$ has classifying map given by
$\begin{pmatrix}
-1 & 0 \\ -2 & 1
\end{pmatrix}$.

\begin{theorem}
Consider the periods of the differential forms $\omega_1$ and $\omega_2$ on $R_c$ as analytic functions of $c=(c_3,c_4)$. Then there is a loop $\Gamma$ in the $c$-space $\C^2\setminus\{(c_3,0):c_3\in\C\}$ such that after analytic continuation along one circuit of $\Gamma$, the period of $\alpha=\left(\frac{i}{c_1}+\frac{-c_3}{c_1(c_2+ic_3)}\right)dc_1$ turns into $-4\pi-\int_\gamma\alpha$ and the period of $\omega_2$ changes the sign.
\end{theorem}
\begin{proof}
The proof is similar to that of the Theorem \ref{monodromy}, and we only sketch the key points and the main differences.

The same transformation as constructed in the proof of the Theorem \ref{monodromy} works also here. Take the point $(1,0)$ in $\{(c_3,0):c_3\in\C\}$ and a circle $\Gamma$ with the center $(1,0)$ and the radius smaller than $\frac{1}{2}$ in the complex plane $c_3=1$. On $\Gamma$, the period of $\alpha$ and $\omega_2$ will be \[\int_\gamma\alpha=\int_\gamma\left(\frac{i}{u-(\frac{1}{2}+\frac{1}{2\sqrt{1+2c_4}})}+\frac{-\sqrt{c_4}}{\sqrt{1+2c_4}}\frac{1}{\sqrt{2u(u-1)}}\frac{1}{u-(\frac{1}{2}+\frac{1}{2\sqrt{1+2c_4}})}\right)du\] and \[\int_\gamma\omega_2=\int_\gamma\left(\frac{\sqrt{1+2c_4}}{c_4\sqrt{c_4}}\frac{u-(\frac{1}{2}+\frac{1}{2\sqrt{1+2c_4}})}{\sqrt{2u(u-1)}}\right)du.\] As shown in Figure \ref{fig_c}, the pole $\frac{1}{2}+\frac{1}{2\sqrt{1+2c_4}}$ of $\alpha$ goes along a closed loop which encircles $1$, when we let $c_4$ go around $\Gamma$ once. By choosing $\gamma$ as a circle with sufficiently large radius, we get \[\int_\gamma\alpha=-2\pi+\frac{-\sqrt{c_4}}{\sqrt{1+2c_4}}\int_\gamma\left(\frac{1}{\sqrt{2u(u-1)}}\frac{1}{u-(\frac{1}{2}+\frac{1}{2\sqrt{1+2c_4}})}\right)du.\] When $c_4$ goes through $\Gamma$ once, the integral $\int_\gamma\left(\frac{1}{\sqrt{2u(u-1)}}\frac{1}{u-(\frac{1}{2}+\frac{1}{2\sqrt{1+2c_4}})}\right)du$ does not change but $\frac{-\sqrt{c_4}}{\sqrt{1+2c_4}}$ changes sign, then $\int_\gamma\alpha$ turns into $-4\pi-\int_\gamma\alpha$. For the same reason, $\int_\gamma\omega_2$ changes the sign after analytic continuation.
\begin{figure}
\centering
\includegraphics[scale=0.7]{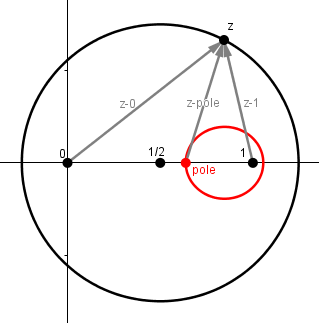}
\caption{Variation of periods in $u$-plane: encompassing the $c_3$-axis.}\label{fig_c}
\end{figure}
\end{proof}

\section{Conclusions and outlooks}

As we have seen, the monodromies are due to the existence of the discriminant loci where the level curves $M_c$ have a singularity and the degeneracy of the level curve along the $c_3$ coordinate axis in $c=(c_3,c_4)$ space. In each case, we construct a loop turning around the singular loci and show that there is nontrivial monodromy. We observed that the monodromy are the same for the both cases, and in fact we can construct a loop encompassing both singular loci with trivial monodromy (i.e., the monodromy effects cancelled due to algebraic additivity of the periods). We tend to believe that the fundamental group of the singular loci
$$ \pi_1(\C^2\backslash (\{(c_3,c_4)\in\C^2\,|\, 1+2c_4c_3^2=0\}\cup \{(c_3,c_4=0)\in \C^2\,|\, c_3\in\C\}))\cong \Z^2,$$
although we could not establish it at present. If this is the case, then the monodromy group would be
$$\Z_2\oplus\Z_2.$$

The next step would be to understand the dynamical meaning of the nontrivial monodromy. To name a few, is there any quantum mechanics implications as suggested by quantum molecular spetroscopy? Is there any physical explanation about the unipotent property of the monodromy matrix?

The monodromy of the complexified Kepler problem shows the potential to the relation with the quantum behavior. Maybe we can benefit more from the complexification point of view, for example the unification of the positive and negative orbits and the regularization problems which we left for the future research.

An even more interesting question would be about complexified $3$-dimensional Kepler problem. As is well known the classical spatial Kepler problem has one more constant of motions the eccentricity or Laplace-Runge-Lenz vector due to the hidden $SO(4)$ symmetry which has far reaching implications about the spectrum of the hydrogen atoms.

Notice that the phase space of the complexified planar Kepler problem as well as any complexified classical mechanical system has beautiful geometry and is a holomorphic symplectic manifold a.k.a. hyperK\"{a}hler manifold (\cite{Bea}). The interaction between the dynamics of the holomorphic Newton equations and the geometry of the underlying phase space would be a fascinating direction to pursue.

\end{document}